\documentclass{article}
\usepackage{amssymb,amsmath,amsthm}
\usepackage{mathrsfs}
\usepackage[plainpages=false]{hyperref}
\usepackage{url}
\usepackage{bbm}   
\usepackage{graphicx}
\usepackage{geometry}
\usepackage{pstricks}
\usepackage{pst-func}
\usepackage{pst-plot}
\usepackage{pst-fill}
\usepackage{subfigure}

\usepackage[inline,marginclue,final,silent]{fixme}

\usepackage{paralist}

\usepackage[all]{xy}
\input xy
\xyoption{all}
\xyoption{arc}
\xyoption{line}

\newcommand{\beq}{\begin{displaymath}}
\newcommand{\eeq}{\end{displaymath}}
\newcommand{\beqn}{\begin{equation}}
\newcommand{\eeqn}{\end{equation}}
\newcommand{\beqa}{\begin{eqnarray*}}
\newcommand{\eeqa}{\end{eqnarray*}}
\newcommand{\beqna}{\begin{eqnarray}}
\newcommand{\eeqna}{\end{eqnarray}}

\newcommand{\eq}[1]{~(\ref{#1})}  

\newcommand{\N}{\mathbb{N}}

\newcommand{\R}{\mathbb{R}}

\newtheorem{prop}{Proposition}[section]
\newtheorem{thm}[prop]{Theorem}

\theoremstyle{definition} 

\title{Quantum correlations in the temporal CHSH scenario}
\author{Tobias Fritz\footnote{This work grew out of joint discussions with Sibasish Ghosh and Tomasz Paterek.}\\
\small{Max Planck Institute for Mathematics}\\
\small{\texttt{fritz@perimeterinstitute.ca}}}

\begin{document}

\maketitle

\begin{abstract}
We consider a temporal version of the CHSH scenario using projective measurements on a single quantum system. It is known that quantum correlations in this scenario are fundamentally more general than correlations obtainable with the assumptions of macroscopic realism and non-invasive measurements. In this work, we also educe some fundamental limitations of these quantum correlations. One result is that a set of correlators can appear in the temporal CHSH scenario if and only if it can appear in the usual spatial CHSH scenario. In particular, we derive the validity of the Tsirelson bound and the impossibility of PR-box behavior. The strength of possible signaling also turns out to be surprisingly limited, giving a maximal communication capacity of approximately $0.32$ bits. We also find a temporal version of Hardy's nonlocality paradox with a maximal quantum value of $1/4$. 
\end{abstract}

\section{Introduction}

Quantum theory displays many counterintuitive features which are in stark contrast to our everyday experiences in the macroscopic world. Possibly the most extreme of these is the collapse of the wavefunction due to measurement; its contentious interpretation has given rise to the \emph{measurement problem}. Obviously, the only possibility to observe and study wavefunction collapse and its entailments is to conduct measurements on the collapsed wavefunction. Therefore, in order to gain a better understanding of what the collapse means and how it occurs, one has to study \emph{repeated} measurements on the same quantum system, both from a theoretical and from an experimental perspective. This should be seen as motivation for our work on temporal quantum correlations. In theories different from orthodox quantum mechanics, for example when wavefunction collapse is not absolutely instanteous~\cite{Pearle}, the properties of temporal correlations are likely to be different from those presented here.

Quantum correlations have mostly been investigated for scenarios of several spacelike separated parties sharing some nonlocal correlations. The simplest situation one can consider here is the Clauser-Horne-Shimony-Holt (CHSH)~\cite{CHSH} scenario: two parties, commonly dubbed \emph{Alice} and \emph{Bob}, each operate with a physical system of their own on which they respectively conduct one of two dichotomic (i.e. two-valued) measurements. Then, on the one hand, quantum theory entails phenomena that cannot be achieved classically: many quantum states that have the property of being \emph{entangled} let Alice and Bob observe correlations between their measurements which cannot be explained by classical models defined in terms of local hidden variables; this non-classicality can be detected by observing violations of the \emph{CHSH inequalities}. These inequalities precisely characterize those correlations having local hidden variable models. Furthermore, \emph{Hardy's nonlocality paradox}~\cite{Har} shows that this feature is not solely a quantitative trait of the joint outcome probabilities: it proves that there also exists a qualitative difference between quantum correlations and the realm of local hidden variable models. On the other hand, it has been found out that there are nevertheless strict limitations on which correlations can be observed with quantum-mechanical systems. The \emph{Popescu-Rohrlich box} (PR-box) is a joint probability distribution that is consistent with the causality principle of \emph{no-signaling}, but yet such a PR box cannot be constructed in a quantum-mechanical world. This can be seen most directly from the \emph{Tsirelson bound}, which specifies the maximal quantum violation of the CHSH inequalities.

In this paper, we study a \emph{temporal} version of the CHSH scenario. We may imagine a single physical system in a laboratory, on which the two experimentalists Alice and Bob can conduct their measurements. However it so happens that their work shifts do not intersect, and Alice leaves the lab before Bob arrives. Now it is known that Alice, during her shift, has measured one of the two $\pm 1$-valued observables $a_1$ or $a_2$, and likewise, Bob will measure one of the two $\pm 1$-valued observables $b_1$ or $b_2$. It is crucial to assume that Alice only conducts one of the two projective measurements $a_1$ and $a_2$, so that she cannot disturb the system and its natural dynamics in any other way. Then which joint probability distributions for the measurement outcomes can possibly arise in this way? In the following sections we answer certain aspects of this question. Just like in the spatial case, we find both fundamental possibilities achievable by such quantum correlations, and fundamental limitations on these quantum correlations. There are analogues of all the spatial phenomena mentioned in the previous paragraph: impossibility of hidden variable models---following~\cite{Lap}, no locality or non-invasiveness assumption is actually needed---, a version of Hardy's paradox which turns out to be stronger than in the spatial scenario, the possibility of signaling in a limited form, impossibility of the PR-box, and the Tsirelson bound. Moreover, although the set of joint probabilities realizable by spatial quantum correlations is strictly contained in the set of joint probabilities realizable by temporal quantum correlations, we find that the set of realizable $\emph{correlators}$ is the same in the temporal case as in the spatial case.

There has been a considerable amount of previous work on the properties of temporal quantum correlations. In particular, the Leggett-Garg inequalities~\cite{LG} characterize the probabilistic hidden variable models for the scenario that one measures two-time correlators between three $\pm 1$-valued observables\footnote{In the standard scenario, these three observables are actually a single observable measured at three different times, but this assumption is not relevant to the argument.}, and it is known that these can be violated quantum-mechanically. In contrast to spacelike separated situations, it is not necessary here to have more than one observable for each ``party'', i.e. at each point in time, since the observables between the different points in time need not commute, leading to specifically quantum phenomena. Very recently, Avis, Hayden and Wilde~\cite{AHW} have classified all tight Leggett-Garg inequalities for the two-time correlators between any number of dichotomic observables as precisely the facets of the cut polytope. Some other relevant references include~\cite{Lap} and~\cite{Bru}.

\section{Joint probabilities in the temporal CHSH scenario}
\label{temporalCHSH}

We start with several statements about temporal correlations between projective quantum measurements of $\pm 1$-valued observables. Then we describe the temporal CHSH scenario, which has been outlined in the introduction, in a little more detail.

\paragraph{Setting the stage.}
\label{stage}
Consider a single quantum system with an underlying Hilbert space $\mathcal{H}$ and dynamics described by the Hamiltonian $H$. Furthermore, we have $\pm 1$-valued, i.e. dichotomic, observables $a$ and $b$, which are hermitian operators on $\mathcal{H}$ with the property
\beq
a^2=\mathbbm{1},\qquad b^2=\mathbbm{1}.
\eeq
Note that we can bring any pair of two-valued observables into this form by relabelling the outcomes as $+1$ and $-1$. Now Alice measures $a$ at time $t_A$ and Bob measures $b$ at time $t_B$. Both measurements are assumed to be perfect projective von Neumann measurements, so that the state collapses to an eigenstate of the corresponding observable upon the measurement. This assumption is relevant for Alice since it limits the way in which her measurement $a$ can influence the system; we will see in paragraph~\ref{CHSHpar} that if we would allow arbitrary generalized measurements (L\"uders measurements) for Alice, then any set of joint outcome probabilities without signaling from Bob to Alice could be modelled even with commuting Kraus operators, i.e. with a classical probabilistic system.

However for Bob, the assumption of projective measurements is not essential: since his post-measurement state does not get measured, this post-measurement state is irrelevant and only his outcome probabilities matter. And concerning these, we can always enlarge the Hilbert space to turn any POVM into a projective measurement while preserving the outcome probabilities.

We take the system to be in the pure initial state $|\psi\rangle$ just before Alice's measurement at time $t_A$. The assumption of a pure initial state is merely for notational convenience, and all following calculations would also apply mutatis mutandis to the case of a mixed initial state. Note also that in case of a mixed initial state described by a density operator $\rho$ on $\mathcal{H}$, we can replace it by a purification $|\psi\rangle$ on $\mathcal{H}\otimes\mathcal{H}'$ for some $\mathcal{H}'$, while replacing the observables $a$ and $b$ by $a\otimes\mathbbm{1}$ and $b\otimes\mathbbm{1}$. This retains all joint outcome probabilities.

When working in the Heisenberg picture, the unitary evolution of the state is trivial, while Bob's observables evolve according to
\beq
b'\equiv e^{-iH(t_B-t_A)}b\, e^{iH(t_B-t_A)}.
\eeq
Since the observable $b$ was arbitrary, the evolved observable $b'$ is also just an arbitrary $\pm 1$-valued observable on $\mathcal{H}$. Hence as far as the existence of quantum-mechanical models for joint probabilities is concerned, the dynamics is irrelevant. In particular, we will choose $H=0$ for simplicity, so that $b'=b$. Then wavefunction collapse is the only ``dynamics'' present in our formalism.

\paragraph{Joint probabilities and correlators.} Now we calculate the joint probabilities in terms of $a$, $b$ and $|\psi\rangle$. For the $\pm 1$-valued observable $a$, the projection operator onto the $+1$-eigenspace and the projection operator onto the $-1$-eigenspace are given by, respectively,
\beq
\frac{\mathbbm{1}+a}{2},\qquad \frac{\mathbbm{1}-a}{2}.
\eeq
and in the same way for $b$. Using the Born rule together with the projection postulate shows that the joint probability for Alice to get the outcome $r\in\{-1,+1\}$ and for Bob to get the outcome $s\in\{-1,+1\}$ reads as
\begin{align}
\begin{split}
\label{jointprobs}
P(r,s)&=\left\langle\psi\left|\frac{\mathbbm{1}+ra}{2}\cdot\frac{\mathbbm{1}+sb}{2}\cdot\frac{\mathbbm{1}+ra}{2}\right|\psi\right\rangle\\[1pc]
&=\tfrac{1}{4}+\tfrac{1}{4}r\langle\psi|a|\psi\rangle+\tfrac{1}{8}s\langle\psi|b|\psi\rangle+\tfrac{1}{8}rs\langle\psi|\{a,b\}|\psi\rangle+\tfrac{1}{8}s\langle\psi|aba|\psi\rangle.
\end{split}
\end{align}
In this expression, $\{\cdot,\cdot\}$ denotes the anticommutator of two operators. $P(r,s)$ is the probability that Alice observes the outcome $r$, multiplied by the probability that Bob gets the outcome $s$ upon measuring the state of the system after state collapse due to Alice's outcome being $r$.

We also consider \emph{correlators}, which are defined as
\begin{align}
\begin{split}
\label{correlation}
C&\equiv\sum_{r,s}rs\: P(r,s)\\
&=P(+1,+1)+P(-1,-1)-P(-1,+1)-P(+1,-1).
\end{split}
\end{align}
Using\eq{jointprobs}, the correlator can be expressed as
\beqn
\label{corrs}
C=\tfrac{1}{2}\langle\psi|\{a,b\}|\psi\rangle
\eeqn
which is intuitive since only the $rs$-term in equation\eq{jointprobs} suggests any kind of correlation between the outcomes. So strangely, even though our scenario has a clear temporal order, the correlators do not depend on who measures first! As far as we can see, this curious property does not generalize to observables with more than two outcomes or to scenarios with more than two parties. 

Note that when we use the term ``correlation'', we simply mean ``specification of joint outcome probabilities for all allowed choices of observables'', while the notion of ``correlator'' refers only to the quantity\eq{corrs}.

\paragraph{The CHSH scenario.}
\label{CHSHpar}
In the CHSH scenario, Alice and Bob both have an independent choice between two observables. While Alice can select either the observable $a_1$ or the observable $a_2$, Bob has the freedom to measure either $b_1$ or $b_2$. For each of the resulting four choices, we obtain a distribution of joint probabilities of the form\eq{jointprobs}. We will use the notation
\beqn
\label{probs}
P(r,s|k,l)
\eeqn
to denote the probability that Alice gets the outcome $r\in\{-1,+1\}$ and Bob gets the outcome $s\in\{-1,+1\}$, given that Alice measures $a_k$ and Bob measures $b_l$. Finally, we will use the notation $C_{kl}$ for the correlator between $a_k$ and $b_l$.

As announced in paragraph~\ref{stage}, it will now be proven that any set of probabilities\eq{probs} has a quantum-mechanical representation in terms of generalized measurements (L\"uders measurements) for Alice, under the assumption that these probabilities satisfy causality in the sense that there is no backward signaling from Bob to Alice. This intuitive condition means that the joint probabilities can be factorized as
\beqn
\label{factor}
P(r,s|k,l)=P_B(s|r;k,l)P_A(r|k)
\eeqn
where $P_A(s|k)$ designates the outcome probabilities for Alice's measurement alone, and these are assumed to be independent of Bob's data $l$ and $s$. On the other hand, Bob's conditional outcome probabilities $P_B(s|r;k,l)$ may well depend on Alice's data in an arbitrary way. Condition\eq{factor} is necessary for the existence of a representation via generalized measurements, since the product representation\eq{factor} is essentially how one would typically calculate the joint probabilities starting from the quantum-mechanical data: first determine Alice's outcome probabilities $P_A(r|k)$ given the initial state $|\psi\rangle$, then calculate Bob's outcome probabilities $P_B(s|r;k,l)$, and finally multiply these two probabilities to obtain the desired result. Bob's probabilities $P_B(s|r;k,l)$ depend on the system's quantum state after Alice's measurement, and this state in turn is determined by $k$, $r$ and $|\psi\rangle$.

Conversely, in order to find a quantum-mechanical representation for an arbitrary such set of probabilities, consider a five-dimensional Hilbert space with orthonormal basis $\{|0\rangle,|1^+\rangle,|1^-\rangle,|2^+\rangle,|2^-\rangle\}$. We take the initial state of the system to be $|\psi\rangle=|0\rangle$. There exist generalized measurements such that the state after Alice's measurement is $|1^+\rangle$ if she measured $a_1$ and obtained a $+1$ outcome, and it is $|1^-\rangle$ if she measured $a_1$ and obtained a $-1$ outcome, and similarly for $|2^+\rangle$ and $|2^-\rangle$. Concretely, one can implement such measurements for example by using the Kraus operators
\beq
V_k^r\equiv \sqrt{P_A(r|k)}\:\left|k^r\rangle\langle 0\right|+\sum_{k',r'}\tfrac{1}{\sqrt{2}}\:\big|k'^{r'}\big\rangle\big\langle k'^{r'}\big|,\quad r\in\{-1,+1\}
\eeq
as describing the measurement of $a_k$. The first term guarantees that the post-measurement state of $V_k^r$ is the desired $|k^r\rangle$ and that the given measurement statistics are reproduced, both on the initial state $|\psi\rangle=|0\rangle$. (The other terms are merely needed for satisfaction of the completeness relation $\sum_{r}V_k^{r\dagger}V_k^r=\mathbbm{1}$.) For Bob, we can choose the two POVMs $\left\{\Pi_1^+,\Pi_1^-\right\}$, $\left\{\Pi_2^+,\Pi_2^-\right\}$ with
\beq
\Pi_l^s\equiv\mathrm{diag}\left(\tfrac{1}{2},\:P(s|+1;1,l),\:P(s|-1;1,l),\:P(s|+1;2,l),\:P(s|-1;2,l)\right)\\
\eeq
as representing the measurements $b_1$ and $b_2$; since Bob's post-measurement state does not get observed, we do not have to specify any Kraus operators implementing these POVMs. By construction, these POVMs reproduce the desired outcome probabilities $P_B(s|r;k,l)$ on the corresponding states $|k^r\rangle\in\{|1^+\rangle, |1^-\rangle, |2^+\rangle, |2^-\rangle\}$. This ends the construction of a quantum-mechanical model with generalized measurements for\eq{factor}. Some final remarks: since neither the initial state nor any post-measurement state is a superposition of basis states, this construction effectively yields a classical stochastic system. The trick in the construction is that Alice's post-measurement state keeps track of both her measurement setting and her outcome. This conditional state collapse to mutually orthogonal states would not be possible if we would only allow projective measurements for Alice.

\paragraph{Temporal hidden variable models.}

Using the assumption of what they called ``macroscopic realism'' and ``non-invasiveness'', Leggett and Garg~\cite{LG} derived an inequality satisfied by temporal correlations in hidden variable models which is violated by certain temporal quantum correlations. Macroscopic realism is the assumption that the system is, at each instant in time, definitely situated in one of several distinct states. This system state determines all measurement outcomes exactly; in this sense, all observables possess preexisting definite values. This is thought to apply to macroscopic objects in particular, hence the name ``macroscopic realism'', or more succinctly ``macrorealism''.

The crucial assumption now is non-invasiveness: this postulates that a measurement does not disturb the state of the system. There is an additional hidden assumption which has been made explicit and dubbed ``induction'' by Leggett~\cite{Real}: it is understood that the state of the system at time $t$ is sufficient information to calculate the outcomes of all future measurements. (In other words, causality only propagates forward in time.) All of these assumptions seem rather natural when dealing with macroscopic systems. In a manner analogous to the spatical case, one can now use these premises to derive (see~\cite{Real}, compare~\cite{Bru}) the temporal CHSH inequality:
\beqn
\label{CHSH}
S_{\textrm{CHSH}}\equiv C_{11}+C_{12}+C_{21}-C_{22}\leq 2.
\eeqn
On the other hand, it is known that this inequality can be violated by certain quantum correlations~\cite{Bru}. This is an exciting area due to promising prospects of using such results for testing the applicability of quantum theory in the macroscopic domain~\cite{exp}.

We will get back to hidden variable models in section~\ref{hardysection}.

\paragraph{Comparison to the spatial scenario.}
In general, the non-invasiveness assumption for hidden variable models is the exact analogue of locality in the spatial case. In both cases, the distribution of joint measurement outcomes is a probabilistic combination (i.e. a convex combination) of a collection of realistic models; a realistic model in turn is described by a hidden variable $\lambda$, constant over space and time, which determines all the outcomes of all possible measurements in a definite way. Therefore, there is absolutely no difference between local hidden variable models in spatial scenarios, and non-invasive hidden variable models in temporal scenarios. 

So the reason that one considers inequalities characterizing hidden variable models for temporal scenarios which are different from those in the spatial case is not that the hidden variable models are different --- they are the same. The reason is that the quantum-mechanical correlations are very different and strongly depend on whether one considers a spatial scenario or a temporal scenario. Although the Leggett-Garg inequality is perfectly valid as a spatial Bell inequality in a three-party scenario, it is not interesting in this case: since there is only one observable per party, no quantum violations are possible, and likewise no violations by more general no-signaling theories.

Let us also note that any set of joint outcome probabilities for a spatial Bell test can also appear in the temporal scenario. Mathematically, this follows from the fact that we recover exactly the spatial joint probabilities by taking $a$ and $b$ in\eq{jointprobs} to operate on separate tensor factors. Physically, this is clear since we can just think of Alice's and Bob's spatially separated quantum systems as a single quantum system, and then simply imagine that Alice conducts her measurement first, with Bob's measurement operating at a later time.

To end the comparison with spatial scenarios, let us recast\eq{corrs} in the following form:

\begin{prop}
While a spatial correlator is given by the expectation value of the tensor products of the observables, a temporal correlator is given by half the expectation value of the anticommutator of the observables:
\beq
\textrm{spatial: }\: C=\langle\psi|a\otimes b|\psi\rangle \quad\longrightarrow\quad \textrm{temporal: }\: C=\tfrac{1}{2}\langle\psi|\{a,b\}|\psi\rangle
\eeq
\end{prop}

\paragraph{The qubit case and beyond.}

As a first example of temporal quantum correlations, we consider a single qubit in the Bloch sphere picture. This case has also been treated in~\cite{Bru}.

Let the system have an initial state given in terms of the Bloch vector $\vec v$. A dichotomic observable is described by a unit vector $\vec{a}\in\R^3$, such that the probability for getting the outcome $r\in\{-1,+1\}$ on the state $\vec{v}$ is given by
\beqn
\label{qubitAlice}
\tfrac{1}{2}(1+r\,\vec{a}\cdot\vec{v})
\eeqn
And in case that the outcome $r$ has been observed, the state has collapsed to $r\,\vec{a}$.

The dynamics of the qubit between $t_1$ and $t_2$ in this representation is specified by a rotation matrix $R\in SO(3)$, such that the state prior to Bob's measurement is $R(r\,\vec{a})=r\,R(\vec{a})$. Then given that Alice obtained the outcome $r$, the probability for Bob to get the outcome $s$ is consequently
\beqn
\label{qubitBob}
\tfrac{1}{2}(1+rs\,\vec{b}\cdot R(\vec{a})).
\eeqn
After multiplying the two expressions\eq{qubitAlice} and\eq{qubitBob} to get the joint probability and summing over $r$ and $s$ with the appropriate sign, the correlator explicitly reads according to the definition\eq{correlation}
\beqn
\begin{split}
C & = \tfrac{1}{2}(1 + \vec a \cdot \vec v) \tfrac{1}{2}(1 + \vec b \cdot R(\vec a)) \\
& + \tfrac{1}{2}(1 - \vec a \cdot \vec v) \tfrac{1}{2}(1 + \vec b \cdot R(\vec a)) \\
& - \tfrac{1}{2}(1 + \vec a \cdot \vec v) \tfrac{1}{2}(1 - \vec b \cdot R(\vec a)) \\
& - \tfrac{1}{2}(1 - \vec a \cdot \vec v) \tfrac{1}{2}(1 - \vec b \cdot R(\vec a)) \\
& = R(\vec a) \cdot \vec b.
\end{split}
\eeqn
So remarkably, this correlator does not depend on the initial state,
which is due to the collapse after Alice has measured, and the structure of the correlator as a particular linear combination of joint probabilities.
This correlator is very similar to the correlator known from maximally entangled two-qubit states
and therefore we can now find the maximal qubit value using simple techniques.
The CHSH quantity then reads:
\begin{align*}
S_{\textrm{CHSH}}^{\textrm{qubit}} &= C_{11}+C_{12}+C_{21}-C_{22}\\
&= R(\vec a_1) \cdot (\vec b_1 + \vec b_2) + R(\vec a_2) \cdot (\vec b_1 - \vec b_2)
\end{align*}
For finding its maximum, note that since the vectors $\vec b$ are normalized, the vectors in the brackets are orthogonal.
Moreover, $|\vec b_1 + \vec b_2|^2 + |\vec b_1 - \vec b_2|^2 = 4$ and so we can introduce
two new orthogonal normalized vectors $\vec b_+$ and $\vec b_-$ such that
$\vec b_1 + \vec b_2 = 2 \cos \alpha \, b_+$ and $\vec b_1 - \vec b_2 = 2 \sin \alpha \, b_-$ for some angle $\alpha$.
Plugging this into the expression for $S_{\textrm{qubit}}$ and optimizing over the $R(\vec{a}_i)$, which are also normalized vectors, yields the Tsirelson bound of $2\sqrt{2}$, which is therefore the maximal value achievable with a qubit. In particular, this violates the bound\eq{CHSH}, confirming that quantum theory cannot be equivalent to a probabilistic hidden variable theory with preexisting values for all observables and repeatable measurements.

All the concrete examples of temporal quantum correlations which we will consider in the following sections are modelled on qubits. So here let us quickly demonstrate that not all quantum correlations in the temporal CHSH scenario can arise from qubit data. Consider a qutrit system with orthonormal basis $\{|0\rangle,|1\rangle,|2\rangle\}$, and the following prescriptions:
\begin{itemize}
\item the initial state $|\psi\rangle=|0\rangle$,
\item $a_1$ measures if the system is in the state $|0\rangle+|1\rangle$,
\item $a_2$ measures if the system is in the state $|0\rangle+|2\rangle$,
\item $b_1$ measures if the system is in the state $|2\rangle$,
\item $b_2$ is any dichotomic observable.
\end{itemize}
This system has the following properties: Alice's outcomes both have probability $1/2$, independent of whether she chooses $a_1$ or $a_2$. But her choice drastically affects Bob's prospects upon measuring $b_1$: when Alice chooses $a_1$, he will definitely observe a $-1$ outcome; however when Alice chooses $a_2$, his outcome will be uniformly random and independent of hers. Such behavior is impossible in a qubit system: one would necessarily need to have $b_1=-\mathbbm{1}$, otherwise Bob's outcome could not be definite after Alice's non-trivial measurement of $a_1$. But then obviously his outcome would also have to be a definite $-1$ when Alice measured $a_2$, which it is not allowed to be. It would be interesting to try and turn this into a dimension witness in the sense of~\cite{DW}.

\section{Correlator space and the Tsirelson bound}

We may ask whether the temporal correlators satisfy the Tsirelson bound generally, or whether this just holds for the case of a qubit system. From the qubit case we know that the Tsirelson bound can be attained; but a priori, some temporal quantum correlations may in principle be so strong that even the Tsirelson inequality
\beqn
\label{tsirelson_bound}
S_{\textrm{CHSH}}=C_{11}+C_{12}+C_{21}-C_{22}\leq 2\sqrt{2}
\eeqn
is violated.

What we mean here by \emph{correlator space} is the set of quadruples
\beq
(C_{11},C_{12},C_{21},C_{22})
\eeq
which can appear as correlators between Alice's and Bob's measurements in a quantum-mechanical world. Recall that the correlators are defined as
\beqn
\label{correlation2}
C_{kl}\equiv \sum_{r,s\in\{-1,+1\}}rs\,P(r,s|k,l)
\eeqn
so that there is a linear map from probability space down to correlator space. Obviously, taking the projection of a point from probability space down to correlator space throws away some data, so specifying the four correlators is not sufficient for knowing the full set of joint probabilities. Yet the correlators contain precious information about the system, for example the maximal violation of the CHSH inequality, and they are also related to the possibility of producing PR-box behavior (see section~\ref{PRboxsection}).

For the remaining part of this section, we will consider the scenario in which Alice has a choice between $m\in\N$ dichotomic observables, while Bob has a choice betweeen $n\in\N$ dichotomic observables. Even in this generality, it is not hard to use the techniques of Tsirelson for showing that, in correlator space, the temporal quantum region coincides with the spatial quantum region. Tsirelson has proven in his paper~\cite{Tsi} that the following three statements are equivalent, for any given matrix of correlators $\left(C_{kl}\right)_{k=1,\ldots,m}^{l=1,\ldots,n}$:
\begin{enumerate}
\item There exists a $C^*$-algebra $\mathcal{A}$ with identity, hermitian elements $a_1,\ldots,a_m,b_1,\ldots,b_n$ and a state $f$ on $\mathcal{A}$ such that for any $k, l$, we have 
\beq
a_kb_l=b_la_k,
\eeq\beq
-\mathbbm{1}\leq a_k\leq\mathbbm{1};\qquad -\mathbbm{1}\leq b_l\leq\mathbbm{1},
\eeq\beq
f(a_kb_l)=C_{kl}.
\eeq
\item There exist Hilbert spaces $\mathcal{H}_a$ and $\mathcal{H}_b$ together with Hermitian operators $a_1,\ldots,a_m\in\mathcal{B}(\mathcal{H}_a)$, $b_1,\ldots,b_n\in\mathcal{B}(\mathcal{H}_b)$ and a density matrix $\rho$ on $\mathcal{H}_a\otimes\mathcal{H}_b$ such that
\beq
a_k^2=\mathbbm{1};\qquad b_l^2=\mathbbm{1}
\eeq\beq
\mathrm{tr}\left(\rho(a_k\otimes b_l)\right)=C_{kl}
\eeq
\item In the Euclidean space of dimension $\min(m,n)$, there exist vectors $x_1,\ldots,x_m,y_1,\ldots,y_n$ such that
\beq
|x_k|\leq 1;\qquad |y_l|\leq 1;\qquad \langle x_k,y_l\rangle=C_{kl}\quad\forall k,l
\eeq
\end{enumerate}

\begin{prop}
These conditions are also equivalent to the following two:
\begin{enumerate}
\item[(a')] There exists a $C^*$-algebra $\mathcal{A}$ with identity, hermitian elements $a_1,\ldots,a_m,b_1,\ldots,b_n$ and a state $f$ on $\mathcal{A}$ such that for any $k, l$ we have
\beq
-\mathbbm{1}\leq a_k\leq\mathbbm{1};\qquad -\mathbbm{1}\leq b_l\leq\mathbbm{1}
\eeq\beq
f\left(\tfrac{1}{2}\left\{a_k,b_l\right\}\right)=C_{kl}
\eeq
\item[(b')] There exists a Hilbert space $\mathcal{H}$ together with Hermitian operators $a_1,\ldots,a_m,b_1,\ldots,b_n\in\mathcal{B}(\mathcal{H})$ and a density matrix $\rho$ on $\mathcal{H}$ such that
\beq
a_k^2=\mathbbm{1};\qquad b_l^2=\mathbbm{1}
\eeq\beq
\mathrm{tr}\left(\rho\cdot\tfrac{1}{2}\left\{a_k,b_l\right\}\right)=C_{kl}
\eeq
\end{enumerate}
\end{prop}

\begin{proof}
We first show that (b)$\Rightarrow$(b'). Given the data as in (b), it is clear that they also satisfy (b') if we take $\mathcal{H}=\mathcal{H}_a\otimes\mathcal{H}_b$ and rename $a_k\otimes\mathbbm{1}$ to $a_k$ and $\mathbbm{1}\otimes b_l$ to $b_l$.

The implication (b')$\Rightarrow$(a') easily follows by choosing $\mathcal{A}=\mathcal{B}(\mathcal{H})$, and $f(x)\equiv\mathrm{tr}(\rho\cdot x)$.

To close the circle of implications, we will now check that (a')$\Rightarrow$(c). But this works in exactly the same way as Tsirelson's own proof~\cite{Tsi} that (a)$\Rightarrow$(c): start with the finite-dimensional vector space defined to be the $\R$-linear span of the $a_k$ and the $b_l$. This vector space carries an inner product, possibly degenerate, which is defined as
\beq
\langle x,y\rangle\equiv f\left(\tfrac{1}{2}\left\{x,y^*\right\}\right)=\mathrm{Re}\:f(y^*x)
\eeq
After quotiening out the null space, this inner product becomes positive definite and produces a Euclidean space such that
\beq
|a_k|^2=\langle a_k,a_k\rangle=f(a_k^2)\leq 1,\qquad |b_l|^2=\langle b_l,b_l\rangle=f(b_l^2)\leq 1,\qquad \langle a_k,b_l\rangle=C_{kl}
\eeq
as required. Now just as in~\cite{Tsi}, all the requirements of (c) are satisfied, except that the dimension of the space has to be at most $\min(m,n)$. This can also be easily achieved by orthogonal projection of the vectors $x_k\equiv a_k$ onto the subspace spanned by the vectors $y_l\equiv b_l$, or in the other way around.
\end{proof}

By\eq{corrs}, we have therefore proven the following result:

\begin{thm}
\label{spatiotemporal}
A matrix of correlators $\left(C_{kl}\right)_{k=1,\ldots,m}^{l=1,\ldots,n}$ can appear as temporal correlations between dichotomic projective measurements on the same system if and only if it can appear as spatial correlations between dichotomic measurements on two spatially separated entangled systems.
\end{thm}

In particular, this implies that the Tsirelson bound\eq{tsirelson_bound} is indeed generally valid in our temporal setting.

\section{Impossibility of PR-box correlations}
\label{PRboxsection}

We say that a \emph{PR-box correlation} is a set of joint probabilities $P(r,s|k,l)$ which has the property that the outcomes $r$ and $s$ are equal if and only if $k=l=2$. This property is equivalent to the requirement that the four correlations\eq{corrs} are given by
\beqn
\label{PRbox}
C_{11}=C_{12}=C_{21}=-1,\qquad C_{22}=+1.
\eeqn
Correlations of this form could be used e.g. to achieve optimal better-than-quantum performance in two-party XOR games (see e.g.~\cite{CHTW}). When the joint probabilities $P(r,s|k,l)$ are assumed to be no-signaling, then this requirement actually fixes all values for the probabilities uniquely; however this does not apply here as our temporal scenario allows signaling from Alice to Bob.

Starting from\eq{corrs}, we now determine when a correlator $C_{kl}$ can have a value of $\pm 1$,
\beq
C_{kl}=\tfrac{1}{2}\langle\psi|\{a_k,b_l\}|\psi\rangle\stackrel{!}{=}\pm 1,
\eeq
which is equivalent to
\beq
\langle\psi|a_kb_l|\psi\rangle+\langle\psi|b_la_k|\psi\rangle=\pm 2.
\eeq
But now since the absolute value of each term is $\leq\! 1$, and becomes $1$ if and only if $|\psi\rangle$ is an eigenstate of the respective operator, it follows that PR-box behavior requires $|\psi\rangle$ to be an eigenstate of the following form:
\beq
b_ka_l|\psi\rangle=a_kb_l|\psi\rangle=(-1)^{(k-1)(l-1)}|\psi\rangle
\eeq
But these equations imply
\beq
\langle\psi|\psi\rangle=\langle\psi|a_1b_1b_1a_2|\psi\rangle=\langle\psi|a_1a_2|\psi\rangle=\langle\psi|a_1b_2b_2a_2|\psi\rangle=-\langle\psi|\psi\rangle
\eeq
which is impossible for any $|\psi\rangle\neq 0$. Therefore, PR-box behavior is impossible even for the temporal quantum correlations which we consider here. We could also have concluded this from theorem~\ref{spatiotemporal}.

\section{Strength of signaling}

In our Bell-test scenario, the backward no-signaling equations
\beq
P(r,-1|k,1)+P(r,+1|k,1)=P(r,-1|k,2)+P(r,+1|k,2)\qquad \forall r,k
\eeq
are still true: the marginal probability governing Alice's measurement cannot possibly depend on the measurement setting of Bob. However the forward no-signaling equations
\beqn
\label{nosig}
P(-1,s|1,l)+P(+1,s|1,l)=P(-1,s|2,l)+P(+1,s|2,l)\qquad \forall s,l
\eeqn
are typically violated, since the choice of measurement for Alice influences the system state after her measurement, and therefore changes the outcome probabilities for Bob. Effectively, what Bob sees is not exactly the initial state $|\psi\rangle$, but $|\psi\rangle$ after undergoing decoherence due to Alice's measurement. It is an interesting question to ask how much the no-signaling equations\eq{nosig} can be violated by our quantum-mechanical setup. This is why we want to look at the deviations from\eq{nosig} and determine how large they can possibly be in a quantum theory. Since each of these four possible quantities involve only one fixed measurement setting $l$ of Bob, we will disregard Bob's choice for the rest of this section, and assume that he simply measures any dichotomic observable $b$. The joint probabilities we then consider are of the form $P(r,s|k)$. Then the two signaling quantities are
\beqn
\begin{split}
S_{+}\equiv& P(+1,+1|1)+P(-1,+1|1)-P(+1,+1|2)-P(-1,+1|2)\\
S_{-}\equiv& P(+1,-1|1)+P(-1,-1|1)-P(+1,-1|2)-P(-1,-1|2).
\end{split}
\eeqn
Due to the total outcome probability for each choice of measurement being $1$, it necessarily holds that $S_++S_-=0$, independent of whether the system is quantum or not. Therefore the interesting question now is, which values of $S_+$ are achievable by quantum mechanics? This is what we are going to answer here.

A priori, $S_+$ can be expected to attain all the values in the interval $[-1,+1]$. The extreme values of $-1$ and $+1$ correspond to perfect signaling in the sense that Bob can definitely tell which measurement Alice had chosen. This can be interpreted as a classical communication channel with a capacity of $1$ bit.

\begin{thm}
A signaling level $S_{+}\in [-1,+1]$ is quantum-mechanical if and only if $|S_+|\leq\frac{1}{2}$.
\end{thm}

\begin{proof}
By\eq{jointprobs}, the signaling quantity $S_+$ can be expressed in terms of the observables and the initial state as
\beqn
\label{Ss}
S_{+}=\tfrac{1}{4}\langle\psi|(a_1ba_1-a_2ba_2)|\psi\rangle\\
\eeqn
where most terms have in fact dropped out. This equation implies
\beq
|S_+|\leq\tfrac{1}{4}|\langle\psi|a_1ba_1|\psi\rangle|+\tfrac{1}{4}|\langle\psi|a_2ba_2|\psi\rangle|
\eeq
Each term within the absolute value brackets in turn can be bounded by $1$, since it is the expectation value of a $\pm 1$-valued observable, so that the bound $|S_+|\leq 1/2$ follows.

Conversely, since the set of allowed for $S_+$ needs to be convex, it is sufficient to show that the values $+1/2$ and $-1/2$ can be attained. For attaining the value $+1/2$, we can choose
\beqn
\label{protocol}
|\psi\rangle=|x+\rangle,\qquad a_1=\sigma_x,\qquad a_2=\sigma_y,\qquad b=\sigma_x,
\eeqn
where a direct calculation shows that this indeed has the required property.
\end{proof}

As was already mentioned briefly, we may also consider the signaling strength in terms of the information which Bob's measurement outcome contains about Alice's choice of setting. This is encoded in the two probabilities
\begin{align}
\begin{split}
P(s|k)&\equiv P(-1,s|k)+P(+1,s|k)\\
&=\tfrac{1}{2}+\tfrac{1}{4}s\langle\psi|b|\psi\rangle+\tfrac{1}{4}s\langle\psi|a_kba_k|\psi\rangle
\end{split}
\end{align}
which define a classical communication channel on the input alphabet $k\in\{1,2\}$ to the output alphabet $s\in\{-1,+1\}$. Since Bob's outcome is only dichotomic, we can equivalently consider the expectation value of his measurement,
\beq
E(s|k)\equiv P(+1|k)-P(-1|k) = \tfrac{1}{2}\langle\psi|b|\psi\rangle+\tfrac{1}{2}\langle\psi|a_kba_k|\psi\rangle,\qquad k\in\{1,2\}
\eeq
and the question then is, which pairs $(E(s|1),E(s|2))$ can occur quantum-mechanically, and how does this bound the classical capacity by which Alice can use her measurements in order to send information to Bob? The answer to this question is given in the following theorem:

\begin{thm}
A pair $(E(s|1),E(s|2))$ can occur quantum-mechanically if and only if
\beq
|E(s|1)-E(s|2)|\leq 1.
\eeq
The maximal communication capacity is $\log_2\left(5/4\right)\approx 0.32\:\mathrm{bits}$, which can be achieved using the qubit protocol\eq{protocol}.
\end{thm}

This result is illustrated in figure~\ref{figEs}.

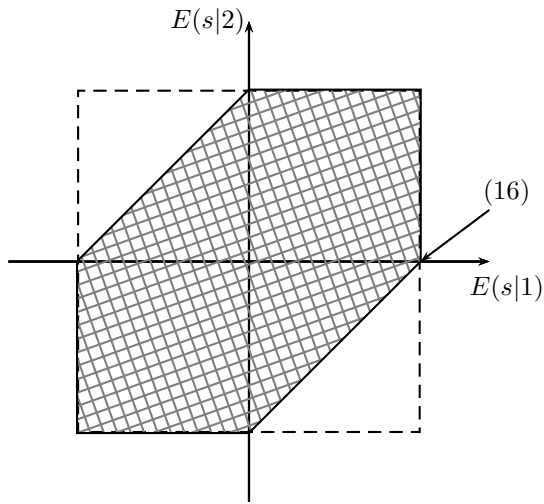
\begin{figure}
\psset{unit=130pt}
\centering{\begin{pspicture}(-.2,-.2)(1.2,1.2)
\psaxes{->}(.5,.5)(-.2,-.2)(1.2,1.2)
\psframe[linestyle=dashed](0,0)(1,1)
\rput(1.25,.42){$E(s|1)$}
\rput(.38,1.2){$E(s|2)$}
\psset{fillstyle=crosshatch,fillcolor=lightgray,hatchcolor=gray,hatchangle=20}
\pscustom{
	\psline(0.5,0)(1,0.5)
	\psline(1,0.5)(1,1)
	\psline(1,1)(.5,1)
	\psline(.5,1)(0,.5)
	\psline(0,.5)(0,0)
	\psline(0,0)(.5,0)}
\psline{->}(1.2,.65)(1,.5)
\rput(1.25,.7){(\ref{protocol})}
\end{pspicture}}
\caption{Possible pairs of $E(s|1)$ and $E(s|2)$ as they can appear in quantum theory. The whole square-shaped box is the whole region of principally allowed values $-1\leq E(s|1),E(s|2)\leq 1$.}
\label{figEs}
\end{figure}

\begin{proof}
Since $E(s|1)-E(s|2)=2S_+$, the constraint $|E(s|1)-E(s|2)|\leq 1$ immediately follows. On the other hand, the qubit protocol\eq{protocol} achieves $E(s|1)=1$, $E(s|2)=0$, which is one of the four non-trivial vertices of the convex quadrangle shown if figure~\ref{figEs}. The other three vertices can be attained by the same protocol after possibly switching $s\leftrightarrow -s$ and $a_1\leftrightarrow a_2$. Now since the quantum region has to be convex, and the quadrangle is the smallest convex set containing its vertices, it follows that ${|E(s|1)-E(s|2)|\leq 1}$ is also sufficient for the existence of a quantum-mechanical model.

Now we get to the capacity statement. Since classical communication capacity is a convex function of the transition probabilities, we know that the maximal capacity is attained at the quadrangle's vertices. Since the four vertices are all simple permutations of the protocol\eq{protocol}, the corresponding channels have equal capacity, and it is sufficient to calculate the capacity achievable by the data\eq{protocol}. A direct calculation shows that the optimal input distribution is a relative frequency of $3/5$ for $a_1$ and $2/5$ for $a_2$, resulting in a mutual information of $\log_2(5/4)\approx 0.32$ bits.
\end{proof}

\section{A temporal version of Hardy's nonlocality paradox}
\label{hardysection}

Hardy's paradox~\cite{Mer} occurs when the joint probabilities have the following properties:
\beqn
\label{hardyparadox}
\begin{split}
P(+1,+1|1,1)=0\\
P(-1,+1|1,2)=0\\
P(+1,-1|2,1)=0\\
P(+1,+1|2,2)>0
\end{split}
\eeqn
This is impossible in any realistic theory where Alice's measurements are non-invasive. We note that the only relevant information contained in hidden variables lies in the preexisting values of all relevant observables. Hence any (stochastic) hidden variable model is given by a statistical mixture of the 16 realistic states
\beqn
\label{realistic}
a_1^\pm a_2^\pm b_1^\pm b_2^\pm
\eeqn
where in this notation (from~\cite{Lap}), each sign stands for the corresponding measurement outcome it determines with certainty, and the four signs are independent of each other. By the assumption $P(+1,+1|2,2)>0$, we know that this statistical mixture contains at least one state of the form
\beq
a_1^\pm a_2^+ b_1^\pm b_2^+.
\eeq
But now due to $P(-1,+1|1,2)=0$, this cannot be one of the two states $a_1^- a_2^+ b_1^\pm b_2^+$. Likewise by $P(+1,-1|2,1)=0$, it cannot be one of the two states $a_1^\pm a_2^+ b_1^- b_2^+$. Therefore, the statistical mixture of realistic states necessarily contains the state
\beq
a_1^+ a_2^+ b_1^+ b_2^+
\eeq
but now this contradicts the assumption $P(+1,+1|1,1)=0$\:! Therefore, the existence of joint probabilities with the property\eq{hardyparadox} exhibits a rather strong form of contextuality. Note that this kind of reasoning applies to a spatial as well as to a temporal Bell test scenario.

In fact, \eq{hardyparadox} is indeed realizable in quantum theory, and it is known that the maximal value for $P(+1,+1|2,2)$ in a spatial scenario is approximately $0.09$~\cite{Mer}. Here we would like to determine the maximal possible value of $P(+1,+1|2,2)$ in the temporal CHSH scenario. Again, since joint probabilities for the temporal case comprise those of the spatial case, the maximal temporally realizable value of $P(+1,+1|2,2)$ has to be at least $0.09$. We will now proceed to show that one can achieve a substantially higher value than this. This shows again that temporal quantum correlations are often stronger than spatial quantum correlations.

\begin{thm}
The maximal value for $P(+1,+1|2,2)$ in the temporal Hardy paradox is $1/4$.
\end{thm}

\begin{proof}
In order for a probability $P(r,s|k,l)$ like\eq{jointprobs} to vanish, one needs that
\begin{itemize}
\item either Alice's outcome $r$ by itself is already impossible to occur, i.e. the other outcome $-r$ occurs with certainty. This means that the initial state is a $-r$-eigenstate of $a_k$, $a_k|\psi\rangle=-r|\psi\rangle$.
\item or Alice's post-measurement state $\frac{\mathbbm{1}+ra_k}{2}|\psi\rangle$ (unnormalized) is such that Bob's outcome is impossible, i.e. it is a $-s$ eigenstate of $b_l$,
\beq
b_l(\mathbbm{1}+ra_k)|\psi\rangle=-s(\mathbbm{1}+ra_k)|\psi\rangle
\eeq
\end{itemize}
Hence, vanishing joint probability is equivalent to
\beqn
\label{jointvanish}
(\mathbbm{1}+sb_l)(\mathbbm{1}+ra_k)|\psi\rangle=0
\eeqn
which can be interpreted as a vanishing amplitude for the two outcomes to occur together. So the vanishing constraints from\eq{hardyparadox} are equivalent to
\begin{align}
&(\mathbbm{1}+b_1)(\mathbbm{1}+a_1)|\psi\rangle=0\label{hardyb1a1},\\
&(\mathbbm{1}+b_2)(\mathbbm{1}-a_1)|\psi\rangle=0\label{hardyb2a1},\\
&(\mathbbm{1}-b_1)(\mathbbm{1}+a_2)|\psi\rangle=0\label{hardyb1a2}.
\end{align}
The qubit protocol
\beq
\label{hardymax}
|\psi\rangle=|x+\rangle,\qquad a_1=-\sigma_x,\qquad a_2=\sigma_y,\qquad b_1=\sigma_y,\qquad b_2=-\sigma_x
\eeq
does indeed satisfy all of these constraints, and it achieves a value of $P(+1,+1|2,2)=1/4$ as promised. The remaining part of the proof is dedicated to showing that this value is optimal.

For $b_1$, the equations\eq{hardyb1a1} to\eq{hardyb1a2} mean that $(\mathbbm{1}+a_1)|\psi\rangle$ has to lie in the $-1$-eigenspace of $b_1$ (since this vector has zero projection onto the $+1$-eigenspace), and similarly that $(\mathbbm{1}+a_2)|\psi\rangle$ has to lie in the $+1$-eigenspace of $b_1$. These eigenspaces are necessarily orthogonal since $b_1$ is hermitian. Hence, given any initial state $|\psi\rangle$ together with $\pm 1$-valued observables $a_1,a_2,b_2$ which satisfy\eq{hardyb2a1}, we can find an observable $b_1$ which also satisfies\eq{hardyb1a1} and\eq{hardyb1a2} if and only if
\beqn
\label{hardya1a2}
(\mathbbm{1}+a_1)|\psi\rangle\perp(\mathbbm{1}+a_2)|\psi\rangle,\qquad\textrm{i.e.}\qquad\langle\psi|(\mathbbm{1}+a_1)(\mathbbm{1}+a_2)|\psi\rangle=0,
\eeqn
holds. So this condition is equivalent to\eq{hardyb1a1} and\eq{hardyb1a2} together and also comprises the case that $|\psi\rangle$ is any eigenstate of $a_1$ or $a_2$.

Now imagine that we have $|\psi\rangle$, $a_1$ and $a_2$ such that\eq{hardya1a2} is satisfied. Then what can we choose for $b_2$ in order to also satisfy\eq{hardyb2a1}? Equation\eq{hardyb2a1} means exactly that $(\mathbbm{1}-a_1)|\psi\rangle$ is contained in the $-1$-eigenspace of $b_2$. When $p$ stands for the projection operator onto the vector $(\mathbbm{1}-a_1)|\psi\rangle$, this means exactly that
\beq
p\leq\frac{\mathbbm{1}-b_2}{2}.
\eeq
(Here, the partial order ``$\leq$'' is the usual partial order on the set of hermitian operators\footnote{recall that $x\leq y$ in this order is defined to mean that $y-x$ is positive semi-definite.}.) On the other hand, we have
\beq
(\mathbbm{1}-a_1)|\psi\rangle\langle\psi|(\mathbbm{1}-a_1)\leq {4p}
\eeq
since the norm of $(\mathbbm{1}-a_1)|\psi\rangle$ is at most ${2}$. Hence, we can conclude from these two inequalities that
\beq
\mathbbm{1}+b_2\leq \mathbbm{2}-2p\leq \mathbbm{2}-{\frac{1}{2}}(\mathbbm{1}-a_1)|\psi\rangle\langle\psi|(\mathbbm{1}-a_1).
\eeq
When plugging this result into the expression for the ``paradoxical'' probability $P(+1,+1|2,2)$, we obtain
\begin{align*}
P(+1,+1|2,2)&=\tfrac{1}{8}\langle\psi|(\mathbbm{1}+a_2)(\mathbbm{1}+b_2)(\mathbbm{1}+a_2)|\psi\rangle\\
&\leq\tfrac{1}{2}\left\langle\psi\left|\left(\mathbbm{1}+a_2\right)\right|\psi\right\rangle-{\tfrac{1}{16}}\langle\psi|(\mathbbm{1}+a_2)(\mathbbm{1}-a_1)|\psi\rangle\langle\psi|(\mathbbm{1}-a_1)(\mathbbm{1}+a_2)|\psi\rangle,
\end{align*}
where it has been used that $(\mathbbm{1}+a_2)^2=2(\mathbbm{1}+a_2)$. But now\eq{hardya1a2} can be applied to evaluate the second term by using
\begin{align*}
\langle\psi|(\mathbbm{1}+a_2)(\mathbbm{1}-a_1)|\psi\rangle=& 2\langle\psi|(\mathbbm{1}+a_2)|\psi\rangle-\langle\psi|(\mathbbm{1}+a_2)(\mathbbm{1}+a_1)|\psi\rangle\\
\stackrel{(\ref{hardya1a2})}{=}& 2\langle\psi|(\mathbbm{1}+a_2)|\psi\rangle.
\end{align*}
Hence we finally end up with
\beq
P(+1,+1|2,2)\leq\tfrac{1}{2}\langle\psi|(\mathbbm{1}+a_2)|\psi\rangle-\tfrac{1}{4}\langle\psi|(\mathbbm{1}+a_2)|\psi\rangle^2
\eeq
which is of the form $x-x^2$ for $x=\frac{1}{2}\langle\psi|(\mathbbm{1}+a_2)|\psi\rangle$. The maximal value of this function is $\frac{1}{4}$, hence the claim is proven.
\end{proof}

\paragraph*{Acknowledgements.} This work was initiated in joint discussions with Sibasish Ghosh and Tomasz Paterek, who therefore have contributed substantially to the present results. It all happened thanks to an invitation by Andreas Winter for the author to visit the Centre for Quantum Technologies in Singapore. Ramon Lapiedra has kindly provided useful comments on an earlier version of this manuscript and endured a long discussion over his work~\cite{Lap}. Four years after this paper was first published, Zhen-Peng Xu pointed out several mathematical mistakes, which have been corrected in this version. Finally, this work would not have been possible without the excellent research conditions within the IMPRS graduate program at the Max Planck Institute.

\bibliographystyle{halpha.bst}
\bibliography{temporal_correlations}

\end{document}